\documentclass[letterpaper]{article}
\usepackage[ruled, lined,vlined,linesnumbered,commentsnumbered]{algorithm2e}
\usepackage{times}
\usepackage[reqno, ]{amsmath}
\usepackage{amssymb}
\usepackage{amsthm}
\usepackage{graphicx}
\usepackage{setspace}
\usepackage{bm}
\usepackage{algpseudocode}
\usepackage{authblk}
\usepackage{aaai17}

\newcommand{\ceil}[1]{\lceil #1\rceil}

\SetKwInOut{Input}{Input}\SetKwInOut{Output}{Output}

\newtheorem{lemm}{Lemma}

\newtheorem{theo}{Theorem}

\algrenewcommand{\Require}{\State \textbf{Input: }}
\algrenewcommand{\Ensure}{\State \textbf{Output: }}

\makeatletter
\newcommand{\rem}[1]{\noindent{\textbf{#1}}}
\newcommand{\rmnum}[1]{\romannumeral #1}
\newcommand{\Rmnum}[1]{\expandafter\@slowromancap\romannumeral #1@}
\makeatother
\newcommand{\brm}[1]{(\rmnum{#1})}
\title{Randomized Mechanisms for Selling Reserved Instances in Cloud\thanks{This work was supported in part by the National Natural Science Foundation of China Grant 61222202, 61433014, 61502449, 61602440 and the China National Program for support of Top-notch Young Professionals.}}

\author[1,2]{Jia Zhang}
\author[3]{Weidong Ma}
\author[3]{Tao Qin}
\author[1,2]{Xiaoming Sun}
\author[3]{Tie-Yan Liu}
\affil[1]{CAS Key Lab of Network Data Science and Technology, Institute of Computing Technology, Chinese Academy of Sciences, Beijing, China.}
\affil[2]{University of Chinese Academy of Sciences, Beijing, China.}
\affil[3]{Microsoft Research, Beijing, China}
\affil[1]{\{zhangjia, sunxiaoming\}@ict.ac.cn, $^2$\{weima,taoqin,tyliu\}@microsoft.com}
\begin{document}
\maketitle
\begin{abstract}
	Selling reserved instances (or virtual machines) is a basic service in cloud computing. In this paper, we consider a more flexible pricing model for instance reservation, in which a customer can propose the time length and number of resources of her request, while in today's industry, customers can only choose from several predefined reservation packages. Under this model, we design randomized mechanisms for customers coming online to optimize social welfare and providers' revenue.

	We first consider a simple case, where the requests from the customers do not vary too much in terms of both length and value density. We design a randomized mechanism that achieves a competitive ratio $\frac{1}{42}$ for both \emph{social welfare} and \emph{revenue}, which is a improvement as there is usually no revenue guarantee in previous works such as \cite{azar2015ec,wang2015selling}. This ratio can be improved up to $\frac{1}{11}$ when we impose a realistic constraint on the maximum number of resources used by each request. On the hardness side, we show an upper bound $\frac{1}{3}$ on competitive ratio for any randomized mechanism.
We then extend our mechanism to the general case and achieve a competitive ratio $\frac{1}{42\ceil{\log k}\ceil{\log T}}$ for both social welfare and revenue, where $T$ is the ratio of the maximum request length to the minimum request length and $k$ is the ratio of the maximum request value density to the minimum request value density. This result outperforms the previous upper bound $\frac{1}{CkT}$ for deterministic mechanisms \cite{wang2015selling}. We also prove an upper bound $\frac{2}{\log 8kT}$ for any randomized mechanism. All the mechanisms we provide are in a greedy style. They are truthful and easy to be integrated into practical cloud systems.
\end{abstract}
\section{Introduction}
Cloud computing is transforming today's IT industry and more and more enterprise customers and personal customers have moved their computational tasks from local devices to cloud. Among all kinds of cloud service models, infrastructure as a service (IaaS) is the most basic one; among many IaaS services, virtual machines/instances are the most basic ones. Two pricing models\footnote{There is another one, the spot pricing, which is not widely adopted by cloud providers.} are adopted to sell virtual instances, the pay-as-you-go model for on-demand instances and the subscription model for reserved instances. The first model charges customers a fixed per-instance-hour rate for their utilization of instance hours. Compared with the first model, the subscription model offers two benefits to customers: (1) Customers enjoy a much lower per-instance-hour price by reserving instances in advance; (2) it is more reliable to reserve some instances before hands for expected future usage, since there may be no on-demand instances available if one goes to the cloud to request on-demand instances in the last minute.

A limitation of the current subscription model in practice is that there are only a few reservation options available to customers. For example, Amazon's EC2 only provides 1-year and 3-year reservations for customers\footnote{This business model is still used when this paper is composed.}. However, many cloud customers demand short-term and flexible reservations. For example, a researcher needs to run many experiments in the last week before the paper deadline of an academic conference and wants to reserve $100$ virtual instances for that week. Certainly, she does not want to reserve for 1 or 3 years\footnote{Actually, this leads to the formation of \emph{(AWS) Reserved Instance Marketplace} where users sell their redundant instances. http://goo.gl/hwzXx9.}. In this work, we study a more flexible pricing model in which the customers can define the reservation option by themselves.

\subsection{Model} \label{mod}
We first formally define the model studied in this work and introduce some notations for further use.

Consider a cloud provider with $C$ reserved instances to be sold to customers. Each customer has a request to reserve some instances. We denote this request as a reservation or a job. Each reservation $j$ is characterized by a 5-tuple $(a_j,d_j,t_j,c_j, v_j)$, where $a_j$ and $d_j$ are earliest start time\footnote{$a_j$ is not the submission time of the reservation. Since we consider reservation, the submission time of the reservation is assumed to be no later than $a_j$.} and deadline respectively, $t_j$ is the reservation length (certainly, $t_j\leq d_j-a_j$), $c_j$ is the number of resources needed by $j$, and $v_j$ is the value that the customer can obtain if $j$ is finished on time. For convenience, let $\rho_j$ be the value density of reservation $j$, i.e., $\rho_j=\frac{v_j}{c_jt_j}$.

Considering that the length and value density of a reservation cannot be arbitrary large/small in practice, we assume
	$\rho_j\in [\rho_{\min}, \rho_{\max}]$
	and $t_j\in [t_{\min}, t_{\max}]$ for any reservation $j$.
	We also denote $k=\frac{\rho_{\max}}{\rho_{\min}}$ and $T=\frac{t_{\max}}{t_{\min}}$. Without loss of generality we assume $\rho_{\min}=t_{\min}=1$. For any reservation set $S$, we denote $v(S)$ as $\sum_{j\in S}v_j$.

We consider the online setting, in which customers come one by one in online style and the cloud provider does not have knowledge about future customers (or reservations). When a customer comes and submits a reservation $j$, the cloud provider needs to immediately and irrevocably decide whether to accept this reservation or not and return a price $p_j$ for this reservation if it is accepted. Immediate response is necessary as customers usually do not have patience to wait for response. Waiting for a long time will affect the user experience. More seriously, a delayed rejection may be a disaster for some tasks if their deadlines are approaching. 

We assume all customers are \emph{rational}, i.e., they are self-interested and always trying to maximize their utilities. For any customer, if her reservation $j$ is rejected, her utility is $0$; otherwise, the utility is $v_j-p_j$. The customer may cheat the provider and misreport her reservation if doing so can increase her utility. Suppose the reported information of $j$ is $(\hat{a}_j, \hat{d}_j, \hat{t}_j, \hat{c}_j, \hat{v}_j)$. For a rational customer, we can safely assume $\hat{a}_j\geq a_j$, $\hat{d}_j\leq d_j$, $\hat{t}_j\geq t_j$ and $\hat{c}_j\geq c_j$; otherwise, the reservation $j$ may terminate unexpectedly and be charged even if it is not completed. We notice that these assumptions are also adopted in \cite{hajiaghayi2005online} and \cite{wang2015selling}.

To avoid the strategic manipulations by the customers, we consider truthful (or strategyproof) mechanisms in this work. A mechanism is truthful if for any customer, she will truthfully disclose her reservation, no matter how other customers behave. This implies that truthful report will maximize the customer's utility.

We focus on designing truthful mechanisms to optimize both the \emph{social welfare} and the provider's \emph{revenue}. For any mechanism $M$, let $J(M)$ stand for the reservations that $M$ accepts. The social welfare gained by $M$ is denoted as $v(M)$, which equals the total value of reservations in $J(M)$. We use $r(M)$ to stand for the revenue obtained from $M$, where $r(M)$ equals $\sum_{j\in J(M)}p_j$.

\subsection{Our Contributions}
To measure the performance of our mechanisms, we employ the concept of competitive ratio. Given any input instance $I$ (a sequence of reservations and corresponding information), suppose $OPT(I)$ is the optimal mechanism. Let $E(v(M(I)))$ stand for the expected social welfare gained by a randomized mechanism $M$. We say $M$ achieves a competitive ratio $c$ for social welfare if and only if $\min_{I} \frac{E(v(M(I)))}{v(OPT(I))}\geq c$. Similarly, we can define competitive ratio for revenue.

In this work, we design randomized mechanisms for selling reserved instances. Our results are summarized as follows. \footnote{Note that the base of all $\log$ terms in this work is $2$ by default. }
\begin{itemize}
	\item [\brm{1}] We first consider a simple case, in which $k\leq 2$ and $T\leq2$. Although the constraint is strong, this case is useful when reservations are not varying too much. For this case, we design a truthful randomized mechanism that achieves a constant ratio $\frac{1}{42}$ for both social welfare and revenue. When we impose a realistic constraint on the number of resources used by each reservation $j$, i.e., $\frac{c_j}{C}\leq \alpha \leq \frac{1}{2}$, the ratio can be improved to $\frac{1-\alpha}{11-\alpha}$. On the hardness side, we prove an upper bound $\frac{1}{3}$ on competitive ratio for all randomized mechanisms.
	\item [\brm{2}] We then extend our mechanism to general $k$ and $T$, and design a truthful mechanism whose competitive ratio is $\frac{1}{42\log\ceil{k}\log\ceil{T}}$ while the upper bound of deterministic mechanism is $\frac{1}{CkT}$ \cite{wang2015selling}. The ratio works for both social welfare and revenue. We also show an upper bound for the general case: no randomized mechanism can achieve a competitive ratio better than $\frac{2}{\log 8kT}.$
\end{itemize}

The mechanisms proposed in this work have two advantages.
\begin{itemize}
\item They can be easily integrated into real system. It is always a concern that randomized mechanisms are hard to be implemented in practice. This is not a problem for our mechanisms. After generating several random parameters, our mechanisms calculate a threshold price for every reservation using the submitted parameters and accept it if its value is larger than the price and we have enough resources for it.
\item They achieve the same performance guarantee for both social welfare and revenue. Previous works in related fields, such as \cite{azar2015ec,wang2015selling,ghodsionline}, usually only care about social welfare, and there is no guarantee for revenue. Although social welfare is an important measurement for mechanism design, clearly revenue is much more important for cloud providers.
\end{itemize}

\section{Related Work}
Our work is related to both online mechanism design \cite{Witkowski2011Trust,Chandra2015Novel} and cloud computing \cite{wang2015selling,zhang2015online}. We only review the most related ones in this part.

In \cite{wang2015selling}, the authors study the same model as us. Truthful deterministic mechanisms are designed under the condition that the number of resources used by each reservation is constrained. For general case, an upper bound $\frac{1}{CkT}$ for deterministic algorithms is provided. As aforementioned, their mechanisms do not provide any guarantee to provider's revenue.

Truthful online mechanisms for reusable or time sensitivity goods have been well studied \cite{friedman2003pricing,porter2004mechanism,lavi2005online,hajiaghayi2005online,Gerding2011Online,Robu2013An,Wu2014A}. Among these works, auction based models are adopted in \cite{friedman2003pricing,lavi2005online,hajiaghayi2005online}. In their models, the seller does not need to make immediate responses to buyers: a buyer needs to wait until the deadline of her reservation, and payment is determined when the deadline passed. In \cite{porter2004mechanism}, the author only considers the one resource case. Except \cite{hajiaghayi2005online}, the other three works only care about the social welfare and there is no guarantee for revenue. Compared to these models, our setting (immediate response with both social welfare and revenue guarantee) is more appropriate for cloud computing.

There are many works that concern online resource allocation for cloud computing \cite{Zaman2012An,zhang2013framework,shi2014anonlinea,mashayekhy2015online,zhang2015online,Wang2015Optimal}. Among these works, the pay-as-you-go model is considered in \cite{mashayekhy2015online}, and authors in \cite{zhang2013framework} propose a framework for cloud resource allocation when agents' value functions are continuous and concave. In \cite{zhang2015online}, a more general scheduling problem is studied. The authors also take both social welfare and revenue into consideration. However, the their competitive ratios are related to the total amount of resources. In \cite{Wang2015Optimal}, the authors propose online strategies to reserve instances without any a priori knowledge of future demands. Those strategies are optimal when considering cost management.

In \cite{lucier2013efficient,jain2015near,azar2015ec}, mechanisms design for scheduling problem with commitments is studied. Those mechanisms proposed either complete a reservation or reject it when there is enough time for this reservation to be completed before deadline. In \cite{jain2015near}, the authors study offline settings and design a near-optimal mechanism with commitments. In \cite{lucier2013efficient}, a heuristic truthful mechanism for online scheduling is proposed, but no formal bound of competitive ratio is given. In \cite{azar2015ec}, the authors follow \cite{lucier2013efficient} and design truthful online mechanisms with a constant competitive when reservations are \emph{slack} enough. 

In this work, the number of resources requested by a reservation is fixed and determined by the customer. However, in some scenarios, the customer only provides the total size (number of resources times length) of her reservation, and it is the cloud scheduler that determines how many resources are allocated to process this reservation. Those reservations are called malleable reservations. There are some works in this fields, such as \cite{carroll2010incentive,kell2014improved}.

All upper bounds in this work are proved by employing Yao's Min-Max Principle \cite{yao1977}. This method is adopted in many works to show upper bounds of randomized algorithms, such as \cite{karp1990optimal,mehta2007adwords}.

\section{Warmup: A Simple Case}
In this section, as a warmup, we investigate a simple case, in which $k=T=2$. That is, the reservations do not vary too much in terms of their value densities and lengths: the maximal length (value density) is twice of the minimal length (value density). We present a randomized mechanism that is truthful and achieves a competitive ratio $\frac{1}{42}$ for both social welfare and revenue.
The mechanism is shown in Algorithm \ref{alg:smpgreedy} and we call it \textsc{Random-Pricing}.
\begin{algorithm}
	\label{alg:smpgreedy}
	\caption{\textsc{Random-Pricing}}
	\Input{A sequence of reservations.}
	Uniformly choose a number from $\{0,1\}$, and let it be $i$.\\
	\While{a reservation $l$ comes online}{
		Set $p_l = \rho_{\min}t_l\cdot\max\left\{\left(\frac{C}{2}\right)^i,c_l \right\}$\\
		\If{$v_l \geq p_l$ and there are enough instances and time for reservation $l$ in $[a_l, d_l]$}{
		Accept $l$ and schedule it as early as possible.\\
		Charge the price $p_l$.
		}
		\textbf{else} Reject $l$.
	}
\end{algorithm}

As we can see, this mechanism is simple and clean. The mechanism first generates a random bit, and for each reservation, it sets a threshold price $p_l$ based on the random bit and the information submitted. The reservation will be filtered and rejected if its value is less than the threshold price. It will be accepted if and only if it passes the filtration and its requested resources are available during the reservation period.

Note that the filtration by setting a random threshold price is critical to guarantee the performance of the mechanism in the worst case. Without the filtration, since the number of resources requested by a reservation can vary from $1$ to $C$, accepting a low-value reservation may exclude a reservation with much more value, which leads to a low utilization rate of the cloud and consequently, a bad performance in the worst case. With the filtration, we can ensure the good utilization of the cloud, which can be intuitively explained as follows.
\begin{itemize}
	\item When $i=1$, we consider the case that all reservations need more than $C/2$ machines. If a reservation $l$ passes the filtration but is rejected, this must be because another accepted reservation has already occupied $l$'s time interval. Because only reservation $j$ with value no less than $\rho_{\min}t_j\cdot\max\{C/2,c_j\}$ can pass filtration, the accepted reservation has a relatively large value. Thus, the cloud is well utilized.
	\item When $i=0$, we consider the case that all reservations need no more than $C/2$ machines. If a reservation $l$ passes the filtration but is rejected, since $c_l\le C/2$, at least half of the total resources has been occupied by other reservations. Thus, the cloud is well utilized.
\end{itemize}

To formally analyze the mechanism, we first define some notations. Let $OPT_0$ stand for the optimal allocation when only those reservations needing no more than $C/2$ resources are taken into consideration, and $OPT_1$ denote the optimal allocation when only those reservations needing more than $C/2$ resources are considered. It is easy to see that $v(OPT_0)+v(OPT_1)\ge v(OPT)$. Let $M$ be the allocation of reservations accepted by \textsc{Random-Pricing} and let  $M_0$ (resp. $M_1$) denote the allocation of reservations when $i=0$ (resp. $i=1$). Let $E(v(M))$ denote the expectation of the total value of reservations accepted by \textsc{Random-Pricing}. Clearly, $E(v(M))=\frac{1}{2}(v(M_0)+v(M_1))$. Let $L_i = J(OPT_i)\setminus J(M_i)$, for $i\in\{0,1\}$.
\begin{theo}
	\label{lemm:ga2b}
	For $k=2$ and $T=2$, the mechanism \textsc{Random-Pricing}
	\begin{enumerate}
		\item [\brm{1}] is truthful;
		\item [\brm{2}] achieves a competitive ratio $\frac{1}{42}$ for social welfare, i.e., $E(v(M))\geq \frac{1}{42}v(OPT)$ and
		\item [\brm{3}] achieves a competitive ratio $\frac{1}{42}$ for revenue.
	\end{enumerate}
	\end{theo}
\begin{proof}
	Recall that we have assumed $\rho_{\min} = t_{\min} = 1$ without loss of generality, thus $\rho_{\max}=t_{\max}=2$.

\noindent{\bf\brm{1} Truthfulness.} Since in our model we assume for any reservation $j$, $\hat{a}_j\geq a_j$, $\hat{d}_j\leq d_j$, $\hat{t}_j\geq t_j$ and $\hat{c}_j\geq c_j$ (see the model formulation part). We only need to show that $j$ has no incentive to report information with $\hat{a}_j> a_j$, $\hat{d}_j < d_j$, $\hat{t}_j > t_j$ or $\hat{c_j}>c_j$. Setting $\hat{a}_j > a_j$ or $\hat{d}_j < d_j$ cannot improve her utility; instead, it increases the risk that $j$ will be rejected, since \textsc{Random-Pricing} will try to find an feasible interval to schedule in $[\hat{a}_j,\hat{d}_j]$. Besides, because the price is set as $\rho_{\min}t_j\cdot\max\left\{\left(\frac{C}{2}\right)^i,c_j \right\}$, a larger value for $t_j$ and $c_j$ can only decrease the utility. At last, since the price does not depend on $v_j$, the customer also has no incentive to cheat on $v_j$. Thus \textsc{Random-Pricing} is a truthful mechanism.

	\noindent{\bf\brm{2} Social Welfare.} The following analysis is based on the fact that if a reservation $j$ is accepted by \textsc{Random-Pricing}, the corresponding customer will always choose to pay the fee and $j$ will be scheduled, as $p_j$ is always no more than $v_j$.

	When analyzing the competitive ratio, we only need to consider the worst case of the mechanism. Thus we can assume\footnote{ If $OPT$ accepts a reservation (e.g., $j$) with value density less than $2$, we can construct another allocation $OPT'$ which accepts exactly the same reservations as $OPT$ except changing the value density of reservation $j$ to $2$. Clearly, $v(OPT')>v(OPT)$, and then we can conduct analysis based on $OPT'$.} all reservations accepted by the mechanism have value density $\rho_{\min}$, and reservations accepted by $OPT$ have value density $2$.
	
	First, we consider the case of $i=1$ in \textsc{Random-Pricing}. 
	For any reservation $l$ in $L_1$, there must be at least one reservation (denoted as $j$) accepted by \textsc{Random-Pricing} that is in conflict with $l$.  Here the \emph{conflict} means the reservation $j$ (accepted by \textsc{Random-Pricing}) occupies at least one resource that is allocated to $l$ in $OPT_1$ in corresponding period. We know the value $v_j$ is at least $t_j\cdot\max\left\{C/2,c_j \right\}$. There may be multiple reservations in $L_1$ that are in conflict with $j$; denote them as a set $F_j$.

	Let $[r, r+t_j]$ be the time interval that \textsc{Random-Pricing} allocates to reservation $j$. All reservations in $F_j$ must be allocated in the interval $[r-2, r+t_j+2]$ by $OPT_1$, otherwise they can't conflict with $j$. In addition, those reservations in $F_j$ are compatible with each other, i.e., they are not in conflict with each other. Thus, the total value of $F_j$ is at most $2C(4+t_j)$, and then
	{\small
	\begin{equation}
		\label{sec3:equ:f1}
		\frac{v_j}{v(F_j)}\geq \frac{t_j\cdot\max\left\{C/2,c_j \right\}}{2C(4+t_j)} \geq \frac{1}{20}.
	\end{equation}
	} The last inequality is obtained as $\frac{t_j}{4+t_j}\geq \frac{1}{5}$. We further have
	{\small
	\begin{equation}
			v(M_1) = \sum_{j\in J(M_1)} v_j
					  \geq \frac{1}{20}\sum_{j\in J(M_1)} v(F_j) \geq \frac{1}{20}v(L_1).
	\end{equation}
	}This also implies $v(L_1) \leq \frac{20}{21}v(OPT_1)$, as $OPT_1=L_1\cup M_1$. Thus we have
			$v(M_1)= v(OPT_1)-v(L_1)\geq \frac{1}{21}v(OPT_1).$

	Second, we consider the case of $i=0$. Similar to the case of $i=1$, if we can show $v(M_0) \geq \frac{1}{20} v(L_0)$, there must be $v(M_0)\geq \frac{1}{21}v(OPT_0)$.
	
	For any $l\in L_0$, suppose it is allocated in time interval $[r,r+t_l]$ by $OPT_0$. Since it cannot be accepted by \textsc{Random-Pricing}, at least $C-c_l+1$ resources are occupied by some reservations in the interval $[r,r+t_l]$ (maybe not the full interval). Denote those reservations as $J_r^{r+t_l}$ and those occupied resources as $M_r^{r+t_l}$. Now we distribute $v_l$ evenly onto the resources in $M_r^{r+t_l}$. As we can see, each reservation $j \in J_r^{r+t_l}$ gets a distributed value no more than $\frac{v_lc_j}{C/2}$. Do the same operation on all the reservations in $L_0$. Consider an arbitrary reservation $j$ in $J(M_0)$. Suppose it is allocated in the interval $[s,s+t_j]$. As we can see, only reservations allocated in interval $[s-2,s+t_j+2]$ by $OPT_0$ have the chance to distribute a value to reservation $j$. As these reservations are compatible with each other, they have a total value of at most $2C(t_j+4)$. Thus the value distributed on $j$ is at most $\frac{2C(t_j+4)c_j}{C/2} = 4(t_j+4)c_j.$ Denote this value as $D_j$. We have
	{\small
	\begin{equation}
		\label{sec3:equ:f2}
		\frac{v_j}{D_j}\geq \frac{t_j\cdot\max\left\{1,c_j \right\}}{4(t_j+4)c_j}\geq \frac{1}{20}.
	\end{equation}
	}
	Thus,
	{\small
	\begin{equation}
		v(M_0)=\sum_{j\in J(M_0)}v_j\geq\frac{1}{20}\sum_{j\in J(M_0)}D_j\geq \frac{1}{20}v(L_0).
	\end{equation}
	}
	Consequently, $v(M_0)\geq \frac{1}{21}v(OPT_0)$.

	In summary, we have
	{\small
	\begin{equation}
		\begin{split}
			&E(v(M)) = \frac{1}{2} v(M_0) + \frac{1}{2} v(M_1)\\
					   &\geq \frac{1}{42}\left(v(OPT_0)+v(OPT_1)\right)
					   \geq \frac{v(OPT)}{42}.
		\end{split}
	\end{equation}
	}

	\noindent{\bf\brm{3} Revenue.} Similar to the above proof for social welfare, we assume all the reservations accepted by \textsc{Random-Pricing} have value density $\rho_{\min}$ for the worst case analysis. Under this assumption, we have $E(v(M))\geq \frac{v(OPT)}{42}$. On the other hand, $\rho_{\min}$ is the value density that used by \textsc{Random-Pricing} to calculate price. This means $E(r(M))\geq \frac{v(OPT)}{42}$.
	It is clear that the optimal revenue will never exceed the optimal social welfare as all customers are rational and the price of a reservation is no larger than its value. Thus, when considering revenue, \textsc{Random-Pricing} also achieves a competitive ratio $\frac{1}{42}$.
\end{proof}
\rem{Remark 1.} Actually, the mechanism \textsc{Random-Pricing} works for any $k,T$ and can achieve a competitive ratio $\frac{1}{8Tk+4k+2}$ for social welfare and revenue. This can be shown by slight modifications on above proof: when $i=1$, Equation (\ref{sec3:equ:f1}) changes to $\frac{v_j}{v(F_j)}\geq \frac{t_j\cdot\max\left\{C/2,c_j \right\}}{kC(2T+t_j)}\geq \frac{1}{4Tk+2k}$, and similarly, Equation (\ref{sec3:equ:f2}) changes to $\frac{v_j}{D_j}\geq \frac{1}{4Tk+2k}$ when $i=0$. Clearly, the \textsc{Random-Pricing} has a better competitive ratio than $\frac{1}{42}$ when $k,T\leq 2$.


\rem{Remark 2.}
In \cite{wang2015selling}, the authors consider the case when $\frac{c_j}{C}$ is bounded. Actually, if we impose a similar constraint that $\frac{c_j}{C}\leq \alpha\leq \frac{1}{2}$, we can design a truthful mechanism to improve the competitive ratio. Consider the following mechanism. When a reservation $l$ comes, we accept it if there are enough time and resources available, and charge the customer $p_l = \rho_{\min}c_lt_l$. We call this mechanism \textsc{Greedy}. Although $\textsc{Greedy}$ is very simple, its performance is guaranteed, as shown in the following theorem. Due to the space limitation, we put the proof in the supplemental material.

\newcommand{\theogreedy}{If $k=T=2$ and $\frac{c_j}{C}\leq \alpha \leq 0.5$ for any reservation $j$, the mechanism \textsc{Greedy} is truthful and achieves a competitive ratio $\frac{1-\alpha}{11-\alpha}$ for both social welfare and revenue. }
\begin{theo}
	\label{theo:greedy}
	\theogreedy
\end{theo}
\subsection{Hardness Analysis }
In this subsection we show an upper bound on competitive ratio for any (randomized) mechanism when $k=T=2$. We first review a powerful tool - Yao's Min-Max Principle, which is usually adopted to analyze the performance of randomized algorithms.

\noindent{\bf Yao's Min-Max Principle}\ \ Given a problem $P$, $\mathcal{A}$ is the set of all the deterministic algorithms to solve $P$. Suppose $\mathcal{I}$ is an input distribution. For any $I\in \mathcal{I}$, let $R(I)$ (resp. $A(I)$) stand for the competitive ratio of a randomized (resp. deterministic) algorithm $R$ (resp. $A$). Then we have $\min_{I\in \mathcal{I}}R(I)\leq \max_{A\in \mathcal{A}}E_{I\sim\mathcal{I}}(A(I))$.

Leveraging Yao's Min-Max Principle, we can get the following result.

\begin{theo}
	\label{theo:2bound}
	For $k=T=2$, no (randomized) mechanism can achieve a competitive ratio better than $\frac{1}{3}$ for social welfare.
\end{theo}
\begin{proof}
	Consider the following 6 bundles of reservations:
	{\small
	\begin{equation*}
		\begin{split}
			B_1 &= \{(2-\varepsilon,3+\varepsilon,1+2\varepsilon,C/2+1, (1+2\varepsilon)(C/2+1))\},\\
			B_2 &= \{(1.5,3.5,2,C/2+1, 2(C/2+1))\},\\
			B_3 &= \{(1.5,3.5,2,C/2+1, 4(C/2+1))\},\\
			B_4 &= \{(0.5,2.5,2,C/2+1,4(C/2+1)),\\
				&\qquad(2.5,4.5,2,C/2+1,4(C/2+1))\},\\
			B_5 &= \{(0.5,2.5,2,C,4C),\,(2.5,4.5,2,C,4C)\},\\
			B_6 &= \{(0,2,2,C,4C),\,(3,5,2,C,4C),\\
			&\qquad(2,3,1,C,2C)\}.
		\end{split}
	\end{equation*}
	}It is easy to check that the reservations in the same bundle are compatible to each other. Now we construct the input distribution $\mathcal{I}$ which contains six input instance $I_1,\cdots,I_6$. In $I_i$, the reservations are submitted exactly in the order $B_1\rightarrow B_2\rightarrow \cdots\rightarrow B_i$ and their submission time is earlier than time $0$. Jobs in the same bundle come at the same time. The input distribution $\mathcal{I}$ is constructed by uniformly choosing one instance from $I_1$ to $I_6$. Any valid deterministic mechanism must belong to one of following patterns.
	\begin{enumerate}
		\item Accept all reservations in only one bundle from $B_1$ to $B_6$;
		\item Choose 2 compatible reservations from $B_4,\,B_5$ and $B_6$ to accept.
	\end{enumerate}
	Besides, we can find that the optimal allocation on input $I_i$ is to accept all reservations in $B_i$.
	Enumerating all deterministic mechanisms, we can easily conclude that the best one is to accept the reservation in $B_1$ for any input instance. Thus, when $\varepsilon\rightarrow 0$ and $C$ is large enough, the expected competitive ratio of the optimal deterministic algorithm is
	\begin{equation}
		\begin{split}
			&\lim_{\varepsilon\rightarrow 0, C\rightarrow \infty}\frac{1}{6}\sum_{i=1}^6\frac{v(B_1)}{OPT(I_i)}\\
			=&\frac{1}{6}\left(1+\frac{1}{2}+\frac{1}{4}+\frac{1}{8}+\frac{1}{16}+\frac{1}{20}\right)
			=\frac{1}{3}-\frac{1}{480}.
		\end{split}
	\end{equation}
	According to Yao's Min-Max Principle, $\frac{1}{3}-\frac{1}{480}$ is an upper bound for any randomized mechanism.
\end{proof}
\section{Mechanism for General $k$ and $T$}
In this section, we consider the general case in which $k$ and $T$ can be arbitrary large. We first present a truthful randomized mechanism with competitive ratio $\frac{1}{42\ceil{\log k}\ceil{\log T}}$ for both social welfare and revenue. We then show that no randomized mechanism can achieve a ratio better than $\frac{2}{\log 8kT}$.
Note that the base of all $\log$ terms is $2$ in this section.

The mechanism, named as \textsc{Binary-Filter}, is shown in Algorithm \ref{alg:algorithmforarbitrary}. It randomly chooses three integers $u$, $v$ and $i$, and sets a price for each reservation based on these three values. Reservations whose values are less than corresponding prices are filtered out. Then rest reservations are scheduled greedily.
\begin{algorithm}
	\label{alg:algorithmforarbitrary}
	\caption{{\textsc{Binary-Filter}}}
	\Input {A sequence of reservations.}
	Uniformly choose an integer $u$ from $[1,\ceil{\log k}]$, and an integer $v$ from $[1,\ceil{\log T}]$.\\
	Uniformly choose an integer $i$ from $\{0,1\}$.\label{line:s}\\
	\While{a reservation $l$ comes online}
	{
		Set $p_l = 2^{u-1}\cdot \max\left\{\left(\frac{C}{2}\right)^i,c_l\right\}\cdot\max\{2^{v-1},t_l\}$.\\
		\If{$l$ can be scheduled in $[a_l, d_l]$ and $v_l\geq p_l$}
		{
			Accept $l$ and schedule it as early as possible.\\
		}
		\textbf{else} Reject $l$.
	}
\end{algorithm}

Let $J_{uv}$ denote the set of reservations whose value densities are located in $[2^{u-1},2^u]$ and lengths are located in $[2^{v-1},2^v]$. Denote the optimal allocation on $J_{uv}$ as $OPT_{uv}$ and denote $M_{uv}$ as the allocation obtained by \textsc{Binary-Filter} when $u$ and $v$ are sampled. We can prove a relationship between $E(v(M_{uv}))$ and $v(OPT_{uv})$ which is shown in Lemma \ref{sec4:lemm:mopt}. As the proof is constructed based on that of Theorem \ref{lemm:ga2b}, we omit it to avoid duplication. Full proof can be found in the supplemental material. 

\newcommand{\lemmmopt}{For any $1\leq u\leq \ceil{\log k}$ and $1\leq v\leq \ceil{\log T}$, $E(v(M_{uv}))\geq \frac{1}{42}v(OPT_{uv})$}

\begin{lemm}
	\label{sec4:lemm:mopt}
	\lemmmopt
\end{lemm}
\begin{theo}
	The mechanism \textsc{Binary-Filter} 
	\begin{enumerate}	
		\item [\brm{1}] is truthful;
		\item [\brm{2}] achieves a competitive ratio $\frac{1}{42\ceil{\log k}\ceil{\log T}}$ for social welfare and
		\item [\brm{3}] achieves a competitive ratio $\frac{1}{42\ceil{\log k}\ceil{\log T}}$ for revenue.
	\end{enumerate}
\end{theo}
\begin{proof}
	\noindent{\bf\brm{1} Truthfulness.} Similar to that of Theorem \ref{lemm:ga2b}, for a customer with jobs $l$ will not cheat on $a_l,d_l$, $c_l$ and $v_l$, thus we only need to show that she will not misreport with $\hat{t}_l>t_l$. When $t_l<2^{j-1}$, $p_l$ is independent of $t_l$; when $t_l\geq 2^{j-1}$, a larger value for $\hat{t}_l$ will increase $p_l$. Thus in either case, the customer will truthfully report $t_l$.

	\noindent{\bf\brm{2} Social Welfare.} 
	According to Lemma \ref{sec4:lemm:mopt},
	{\small
	\begin{equation}
		\begin{split}
			E(v(M)) &= \frac{1}{\ceil{\log k}\ceil{\log T}}\sum_{i=1}^{\ceil{\log k}}\sum_{j=1}^{\ceil{\log T}} E(v(M_{uv}))\\
			&\geq \frac{1}{42\ceil{\log k}\ceil{\log T}}\sum_{i=1}^{\ceil{\log k}}\sum_{j=1}^{\ceil{\log T}} v(OPT_{uv})\\
			&\geq \frac{v(OPT)}{42\ceil{\log k}\ceil{\log T}}.
		\end{split}
	\end{equation}
	}
	
	\noindent{\bf\brm{3} Revenue.} The proof is similar to that of Theorem \ref{lemm:ga2b}.
\end{proof}

\rem{Remark 3.} It has showed that no deterministic algorithm could achieve a competitive ratio better than $\frac{1}{kCT}$ \cite{wang2015selling}. Clearly, in terms of the worst case, \textsc{Binary-Filter} achieves much better performance than any deterministic algorithms for large $k$, $C$ and $T$.

\rem{Remark 4.} If $\frac{c_j}{C}\leq\alpha\leq \frac{1}{2}$ for all $j$, we can modify the \textsc{Binary-Filter} mechanism by removing line 2 in Algorithm 2 and setting $p_l=2^{u-1}c_l\max\{2^{v-1},t_l\}$. This modified mechanism achieves a competitive ratio $\frac{1-\alpha}{(11-\alpha)\ceil{\log k}\ceil{\log T}}$ for social welfare and revenue.

%
\subsection{Hardness Analysis} \label{sec:upperofrandomized}

\begin{theo}
	\label{theo:generalbound}
	For general $k$ and $T$, no (randomized) mechanism can achieve a competitive ratio better than $\frac{2}{\log 8kT}$ for social welfare.
\end{theo}

\begin{proof}
	We prove this theorem by using Yao's Min-Max Principle as for the $2$-bounded case. To use the Yao's Min-Max Principle, we should construct an input distribution $\mathcal{I}$. Without loss of generality, we assume that $C$ is an even integer. Given two positive integers $m$ and $n$, we first construct $m+n+2$ bundles of reservations $B_1,\cdots,B_{m+n+2}$:
{\tiny
	\begin{equation*}
		B_i = \left\{\begin{aligned}
			&\{(2^{n}-2^{i-1}, 2^{n}+2^{i-1}, 2^i, C/2+1, 2^{i-1}(C+2))\},
		   	1\leq i\leq n,&	\\
			&\{(2^{n-1}, 2^{n}+2^{n-1}, 2^{n}, C/2+1, 2^{i-1}(C+2))\},
		   	n+1\leq i\leq n+m,&\\
			&\{(2^{n-1}, 2^{n}+2^{n-1}, 2^{n}, C, 2^{i-1}C)\},
			i = m+n+1,&\\
			&\{(0, 2^{n}, 2^{n}, C, 2^{i-2}C),(2^{n}, 2^{2n}, 2^{n}, C, 2^{i-2}C)\},
			i = m+n+2.&
		\end{aligned}\right.
	\end{equation*}
}
	We can find that for the first $n$ bundles, each reservation has a value density $1$ and the running time changes w.r.t. $i$. For $B_{n+1}$ to $B_{n+m}$, the running time of each reservation keeps unchanged, i.e., $2^n$, but the value density increases w.r.t. $i$. Furthermore, only the reservations in the last two bundles use $C$ resources. It is easy to check that $T=2^{n-1}$ and $k=2^{m}$ for the reservations in these bundles.
	Most importantly, we can find that reservations in different bundles are conflict with each other, as all reservations are tight and valid time intervals of reservations in different bundles overlap with each other. These properties are the preliminary of following analysis.

	The input distribution $\mathcal{I}$ contains $n+m+2$ input instances $I_1,\cdots,I_{n+m+2}$. In the input instance $I_i$, reservations come exactly as the order $B_1\rightarrow B_2 \rightarrow \cdots\rightarrow B_i$. All reservations are submitted before time 0 and reservations in the same bundle have the same submission time. Each input instance is drawn uniformly from the distribution $\mathcal{I}$.

	Suppose $A$ is an arbitrary deterministic algorithm. For any input instance $I$, if $A$ accepts a reservation in $B_j$, it cannot accept any reservation in other bundles as reservations in different bundles are conflict with each other. As $A$ is deterministic, we can conclude that $A$ will always accept the reservations in $B_j$ (if the input instance contains $B_j$) or accept nothing (if the input instance does not contain $B_j$) for any input instance.

	Clearly, when $C$ is large enough, the competitive ratio of $A$ on input instance $I_i$ is
	{\small
	\begin{equation}
		A(I_i) = \left\{
		\begin{aligned}
			&0,& &i < j\\
			&\frac{2^{j-1}}{2^{i-1}}.& &j\leq i\leq m+n+2&
		\end{aligned}\right.
	\end{equation}
	}
	Thus we have
	{\small
	\begin{equation}
		\begin{split}
			&E_{I\sim \mathcal{I}}(A(I)) = \frac{1}{n+m+2}\sum_{i=1}^{n+m+2} A(I_i)\\
			&=\frac{1}{n+m+2}\sum_{i=j}^{n+m+2}\frac{1}{2^{i-j}}
			=\frac{1}{n+m+2}\left(2-\frac{1}{2^{n+m+2-j}}\right).\\
		\end{split}
	\end{equation}
	} As we can see, when $A$ selects the reservation in $B_1$, that is $j=1$, $E_{I\sim\mathcal{I}}(A(I))$ is maximized. That is to say
	{\small
	\begin{equation}
		\begin{split}	
			\max_{A\in \mathcal{A}}E_{i\sim\mathcal{I}}(A(I))&= \frac{1}{n+m+2}\left(2-\frac{1}{2^{n+m+1}}\right)\\
			&\leq \frac{2}{n+m+2}
			=\frac{2}{\log 8kT}.
		\end{split}
	\end{equation}
	}
The last equality is obtained as $n=\log T+1$ and $m=\log k$.

Applying the Yao's Min-Max Principle, we finish the proof.
\end{proof}

\rem{Remark 5.} Theorem \ref{theo:generalbound} does not imply Theorem \ref{theo:2bound}, as the input distribution constructed in the proof of Theorem \ref{theo:generalbound} can only produce an upper bound at least $0.3875$ for $k=T=2$, which is weaker than $\frac{1}{3}$ in Theorem \ref{theo:2bound}.
\section{Future Work}
In this work, we have designed randomized mechanisms for instance reservation in cloud. Our mechanisms are simple and truthful, and have performance guarantee for both social welfare and revenue. There are several directions to explore in the future.

First, the competitive ratios of our mechanisms do not match the upper bound we provided. It is interesting to study whether there exist some methods to narrow down this gap between the upper bound and the lower bound.

Second, as mentioned in the introduction, malleable reservations, which allow the mechanism to decide the number of resources to allocate to a reservation, have been studied in other settings. It is interesting to introduce malleable reservations into instance reservation in cloud and design mechanisms to optimize social welfare and revenue.
\bibliographystyle{aaai}
\bibliography{ref}
\appendix

\noindent{\large\textbf{Supplemental Material}}
\section{Omitted Proofs\label{app:a}}
\noindent\textbf{Theorem \ref{theo:greedy}.} \textsl{\theogreedy}
\begin{proof}(\textsl{Sketch})
	The competitive ratio can be proved by slightly modifying the proof of the case $i=0$ in Theorem \ref{lemm:ga2b}. Here we only provide a proof sketch. 
Let $M$ denote the allocation of \textsc{Greedy} and $L$ denote the set of jobs accepted by the optimal allocation  $OPT$ but not accepted by the \textsc{Greedy}. For a job $l\in L$, $v_l$ is evenly distributed on machines in $M_r^{r+t_l}$. As $c_l/C\leq \alpha$, we know  $|M_{r}^{r+t_l}|\geq C-c_l\geq C(1-\alpha)$. Thus, for any job $j\in J_r^{r+t_l}$, it gets a value no larger than $\frac{v_lc_j}{C(1-\alpha)}$ from $l$. Similar to the proof of Theorem \ref{lemm:ga2b} for the remaining part, we have $\frac{v_j}{D_j}\geq \frac{1-\alpha}{10}$ for $j\in M$. Consequently, it can be obtained that $$v(M)\geq \frac{1-\alpha}{10}v(L).$$ 

If $v(L)\leq \frac{10}{11-\alpha}v(OPT)$, we get this theorem. If $v(L)>\frac{10}{11-\alpha}v(OPT)$, we have $$v(M))\geq \frac{1-\alpha}{10}v(L)\geq \frac{1-\alpha}{11-\alpha}v(OPT).$$
	
Similarly, one can prove the competitive ratio $\frac{1-\alpha}{11-\alpha}$ for revenue. The truthfulness of \textsc{Greedy} can be proved similarly to \textsc{Random-Pricing}.
\end{proof}

\noindent\textbf{Lemma \ref{sec4:lemm:mopt}.} \textsl{\lemmmopt}
\begin{proof}
	For jobs in $J_{uv}$, define $OPT_0$ as the optimal allocation when we only consider those jobs need no more than $C/2$ machines, and define $OPT_1$ as that when we only consider jobs need more than $C/2$ machines. We define and $M_0$ (resp. $M_1$) as the allocation obtained by our mechanism when $i=0$ (reps. $i=1$) given $u$ and $v$ are sampled. We also define $L_i=J(OPT_i)\setminus J(M_i)$ for $i\in \{0,1\}$. Now we have the same notations with the proof of Theorem \ref{lemm:ga2b}. In fact, we can show that Equation (\ref{sec3:equ:f1}) and (\ref{sec3:equ:f2}) are still held, if we can show $v(M_{uv})\geq \frac{1}{42}v(OPT_{uv})$ by the proof of Theorem \ref{lemm:ga2b}, we can show this lemma by simple following the proof of Theorem \ref{lemm:ga2b}. Firstly we consider Equation (\ref{sec3:equ:f1}). As any job in $J_{uv}$, its length is at most $2^{v}$ and its value density is at most $2^u$, Equation \ref{sec3:equ:f1} becomes 
	$$\frac{v_j}{v(F_j)}\geq\frac{2^{u-1}\max\{C/2,c_j\}\max\{2^{v-1},t_v\}}{2^uC(2^{v+1}+t_j)}\geq \frac{1}{20}.$$
	The last inequality is obtained as $c_j > C/2$ and the term $\frac{\max\{2^{v-1},t_j\}}{2^{v+1}+t_j}$ is minimized when $t_j=2^{v-1}$.
	Similarly, when considering Equation \ref{sec3:equ:f2}, we also have $\frac{v_j}{D_j}\geq \frac{1}{20}$. Thus follow the proof of Theorem \ref{lemm:ga2b}, we have $E(v(M_{uv}))\geq \frac{1}{42}v(OPT_{uv})$.
\end{proof}
\end{document}